\documentclass[11pt,fleqn,letterpaper]{article}
\usepackage{amsmath,amssymb,amsthm}
\usepackage{fullpage}
\usepackage{xspace}
\usepackage{multirow}
\usepackage{ifpdf}

\usepackage[ruled,vlined,linesnumbered]{algorithm2e}

\usepackage{fancyhdr}
\usepackage[hypertex]{hyperref}

\newcommand{\REMOVE}[1]{}

\newcommand{\ef}{\frac{1}{\epsilon}}

\newtheorem{theorem}{Theorem}[section]
\newtheorem{lemma}[theorem]{Lemma}

\newtheorem{corollary}[theorem]{Corollary}

\theoremstyle{definition}
\newtheorem{definition}[theorem]{Definition}

\newcounter{this-list}

\newcounter{par-list}
\newlength{\parlistlength}

\begin{document}


\title{Efficient Computation of Distance Sketches in Distributed Networks}
\author{Atish Das Sarma\\Google Research\\Mountain View, CA
\and Michael Dinitz\\Weizmann Institute of Science\\Rehovot, Israel
\and
Gopal Pandurangan\\Nanyang Technological University\\Singapore}

\maketitle

\begin{abstract}
Distance computation (e.g., computing shortest paths) is one of the most fundamental primitives used in communication networks. The cost of effectively and accurately computing pairwise network distances can become prohibitive in large-scale networks such as the Internet and Peer-to-Peer (P2P) networks. To negotiate the rising need for very efficient distance computation at scales never imagined before, approximation techniques for numerous variants of this question have recently received significant attention in the literature. Several different areas of theoretical research have emerged centered around this problem, such as metric embeddings, distance labelings, spanners, and distance oracles. The goal is to preprocess the graph and store a {\em small} amount of information such that whenever a query for any pairwise distance is issued, the distance can be well approximated  (i.e., with small {\em stretch}) very quickly in an online fashion. Specifically, the pre-processing (usually) involves storing a small {\em sketch} with each node, such that at query time only the sketches of the concerned nodes need to be looked up to compute the approximate distance.

Techniques derived from metric embeddings have been considered extensively by the networking community, usually under the name of \emph{network coordinate systems}.  On the other hand, while the computation of distance oracles has received considerable attention in the context of web graphs and social networks, there has been little work towards similar algorithms within the  networking community.
In this paper, we present the first theoretical study of distance sketches derived from distance oracles in a distributed network. We first present a fast distributed algorithm for computing approximate distance sketches, based on a distributed implementation of the distance oracle scheme of  [Thorup-Zwick, JACM 2005].  We also show how to modify this basic construction to achieve different tradeoffs between the number of pairs for which the distance estimate is accurate, the size of the sketches, and the time and message complexity necessary to compute them.  These tradeoffs can then be combined to give an efficient construction of small sketches with provable average-case as well as worst-case performance.  Our algorithms use only small-sized messages and hence are suitable for bandwidth-constrained networks, and can be used in various networking applications such as topology discovery and construction, token management, load balancing, monitoring overlays, and several other problems in distributed algorithms.
\end{abstract}

\section{Introduction}

A fundamental operation on large networks is finding shortest paths between pairs of nodes, or at least finding the lengths of these shortest paths. This problem is not only a common building block in many algorithms, but is also a meaningful operation in its own right.
In a distributed network such as a large peer to peer network, this may be useful in search, topology discovery, overlay creation, and basic node to node communication.  However, given the large size of these networks, computing shortest path distances can, if done naively, require a significant amount of both time and network resources.
As we would like to make distance queries in real time with minimal latency, it becomes important to use small amounts of resources per distance query.

One approach to handle online distance requests is to perform a one-time offline or centralized computation. A straightforward brute force solution would be to compute the shortest paths between all pairs of nodes offline and to store the distances locally in the nodes.  Once this has been accomplished, answering a shortest-path query online can be done with no communication overhead; however, the local space requirement is quadratic in the number of nodes in the graph (or linear if only shortest paths from the node are stored).   For a large network containing millions of nodes, this is simply infeasible.
An alternative, more practical approach is to store some auxiliary information with each node that can facilitate a quick distance computation online in real time. This auxiliary information is then used in the online computation that is performed for every request or query. One can view this auxiliary information as a {\em sketch} of the neighborhood structure of a node that is stored with each node. Simply retrieving the sketches of the two nodes should be sufficient to estimate the distance between them. Three properties are crucial for this purpose: first, these sketches should be reasonably small in size so that they can be stored with each node and accessed for any node at run time. Second, there needs to be a simple algorithm that, given the sketches of two nodes, can estimate the distance between them quickly. And third, even though the computation of the sketches is an offline computation,
this cost also needs to be accounted for (as the distance information or network itself changes frequently, and this would require altering the sketches periodically).



Sketches for the specific purpose of distance computation in communication networks have been referred to as \emph{distance labelings} (by the more theoretical literature) and as \emph{network coordinate systems} (by the more applied literature). There has been a significant amount of work from a more theoretical point of view on the fundamental tradeoff between the size of the sketches and the accuracy of the distance estimates they give (see e.g.~Thorup and Zwick~\cite{TZ05}, Gavoille et al.~\cite{GPPR04}, Katz et al.~\cite{KKKP04}, and Cohen et al.~\cite{CFIKP09}).  However, all of these papers assumed a \emph{centralized} computation of sketches, so are of limited utility in real distributed systems.  In the networking community there has also been much work on constructing good network coordinate systems, including seminal work such as the Vivaldi system~\cite{Vivaldi} and the Meridian system~\cite{WSS05}.  While this line of work has resulted in almost fully functioning systems with efficient distributed algorithms, the theoretical underpinning of such systems is lacking; most of them can easily be shown to exhibit poor behavior in pathological instances.  The main exception to this is Meridian, which has a significant theoretical component.  However, it assumes that the underlying metric space is ``low-dimensional", and it is easy to construct high-dimensional instances on which Meridian does poorly.

We attempt to move the theoretical line of research slightly closer to practice by designing efficient distributed algorithms for computing accurate distance sketches.  We give algorithms with bounded round and message complexity in a standard model of distributed computation (the CONGEST model~\cite{peleg}) that compute sketches that give distances estimates provably close to the actual distances.  In particular, we engineer a distributed version of the seminal centralized algorithm of Thorup and Zwick \cite{TZ05}. Their algorithm computes sketches that allow us to approximate distances within a factor $2k-1$ (this value is known as the \emph{stretch}) by using sketches of size $\tilde O(k \cdot n^{1/k})$, for any integer $k \geq 1$.  Up to the factor of $k$ in the size, this is known to be a tight tradeoff in the worst case assuming a famous conjecture by Erd\H{o}s~\cite{TZ05}.  Note that this achieves its minimum size at $k=\log n$, giving an $O(\log n)$-factor approximation to the distances using sketches of size $O(\log^2 n)$. We further extend these results to give a distributed algorithm based on the centralized algorithms of Chan et al.~\cite{CDG06} that computes sketches with the same worst-case stretch and almost the same size, but with provably better \emph{average} stretch.  To the best of our knowledge, this is the first theoretical analysis of distributed algorithms for computing distance sketches.  Our work can also be viewed as an efficient computation of local node-centric views of the global topology, which may be of independent interest for numerous different applications (cf. Section \ref{sec:applications}).


\subsection{Our Contributions}

Our main contributions are new distributed algorithms for various types of distance sketches.  Most of the actual sketches have been described in previous work~\cite{TZ05,CDG06}, but we are the first to show that they can be efficiently constructed in a distributed network.  While we formally define the model and problem in Section~\ref{sec:model}, at a high level we assume a synchronous distributed network in which in a single communication ``round"  every node can send a message of up to $O(\log n)$ bits (or $1$ word) to each of its neighbors.  Every edge has a nonnegative weight associated with it, and the distance between two nodes is the total weight of the shortest path (with respect to weights) between them.  We begin in Section~\ref{sec:main} by giving a distributed algorithm that constructs Thorup-Zwick distance sketches~\cite{TZ05} efficiently:

\begin{theorem} \label{thm:TZ_intro}
For any $k \geq 1$, there is a distributed algorithm that takes $O(k n^{1/k}S \log n)$ rounds and $O(k n^{1/k} S |E| \log n)$ messages, after which with high probability every node has a sketch of size at most $O(kn^{1/k}\log n)$ words that provides approximate distances up to a factor of $2k-1$.
\end{theorem}

The value $S$ in this theorem is known as the \emph{shortest-path diameter}~\cite{khan-disc}, and is (informally) the maximum over all ${n \choose 2}$ shortest paths (where ``short" is determined by weight) of the number of hops on the path.  (In an unweighted network, $S$ is the same as the network diameter $D$, hence $S$ can be thought
of as a generalization of $D$ in a weighted network). Note that this is essentially a lower bound on any distance computation.

In Section~\ref{sec:slack} we show how to extend the techniques of Chan, Dinitz, and Gupta~\cite{CDG06} and combine them with the techniques of Section~\ref{sec:main} to give sketches with ``slack".  Informally, a sketch has $\epsilon$-slack if the stretch factor (i.e.~the distance approximation guarantee) only holds for a $(1-\epsilon)$-fraction of the pairs, rather than all pairs.  While this is a weaker guarantee, since some pairs have no bound on the accuracy of the distance estimate at all, both the size of the sketches and the time needed to construct them become much smaller.  For example, when $\epsilon$ is a constant (even a small constant) we can construct constant stretch sketches in only $O(S \log^2 n)$ rounds.

\begin{theorem}
For any $\epsilon > 0$ and $1 \leq k \leq O(\log \frac{1}{\epsilon})$, there is a distributed sketching algorithm that gives sketches with size at most $O(k \left(\frac{1}{\epsilon} \log n\right)^{1/k} \log n)$ words and stretch $8k-1$ with $\epsilon$-slack that completes in at most $O\left(k S \left(\frac{1}{\epsilon} \log n\right)^{1/k} \log n\right)$ rounds and $O\left(k S |E|\left(\frac{1}{\epsilon} \log n\right)^{1/k} \log n\right)$ messages.
\end{theorem}

Finally, in Section~\ref{sec:average} we extend the slack techniques further, allowing us to efficiently construct sketches with the same worst-case stretch as in Theorem~\ref{thm:TZ_intro} (with $k = O(\log n)$) but with average stretch only $O(1)$:

\begin{theorem}
There is a distributed sketching algorithm that gives sketches of size $O(\log^4 n)$ with $O(\log n)$-stretch and $O(1)$ average stretch that completes in at most $O(S \log^4 n)$ rounds and at most $O(S |E| \log^4 n)$ messages.
\end{theorem}

\section{Related Work and Model}

\subsection{Applications and Related Works}
\label{sec:applications}


Applications of approximate distance computations in distributed networks include token management~\cite{IJ90, BBF04,CTW93}, load balancing~\cite{KR04}, small-world routing~\cite{K00}, and search~\cite{ZS06,AHLP01,C05,GMS05,LCCLS02}. Several other areas of distributed computing also use distance computations in a crucial way; some examples are information propagation and gathering~\cite{BAS04,KKD01}, network topology construction~\cite{GMS05,LawS03,LKRG03}, monitoring overlays~\cite{MG07}, group communication in ad-hoc network~\cite{DSW06}, gathering and dissemination of information over a network~\cite{AKL+79}, and peer-to-peer membership management~\cite{GKM03,ZSS05}.

The most concrete application of our algorithms is to quickly computing approximate shortest path distances in networks, i.e.~the normal application of network coordinate systems. In particular, in weighted networks, after using our algorithms to preprocess the network and create distance sketches we can compute the approximate distance between any two nodes in at most $O(D)$ times the size of the sketch rounds (where $D$, the hop-diameter, is the maximum over all pairs of nodes of the minimum number of hops between the nodes) by simply exchanging the sketches of the two nodes.  On the other hand, note that any distance computation without using preprocessing (say, Dijkstra's algorithm, Bellman-Ford, or even a simple network ping to obtain the round-trip time) will take at least $\Omega(S)$ rounds, where $S$ is  the shortest path diameter.  This is less than ideal since $S$ can be as  large as $n$, the number of nodes in the networks,  whereas $D$, the hop-diameter can be, and typically is, much smaller.
 Therefore our sketches yield improved algorithms for pairwise weighted-distance computations. Moreover,
 in networks such as P2P networks and overlay networks, using our algorithms a node can compute distances (number of hops in the overlay) in  {\em constant} times size of sketch rounds  if it simply knows the IP address of the other node: it can directly contact the other node using its IP address  and ask for its sketch.
Thus these  sketch techniques can be  very relevant and applicable even for unweighted distance computations.

Probably the closest results to ours are from the theory behind Meridian~\cite{WSS05}, which is based on a modification of the ``ring-of-neighbors" theoretical framework developed by Slivkins~\cite{S05b,S07} to prove theoretical bounds.  However, there are some substantial differences.  For one, the bounds given by these papers (including~\cite{WSS05}) are limited to special types of metric spaces: those with either bounded doubling dimension or bounded growth.  Our bounds hold for all (weighted) graphs (with, of course, weaker guarantees on the stretch).  Furthermore, the distributed framework used by Slivkins is significantly different from the standard CONGEST model of distributed computation that we use, and is based on being able to work in the metric completion of the graph.  This means, for example, that the algorithms in~\cite{S07} have the ability to send a unit-size packet between \emph{any} two nodes in $O(1)$ time (and the algorithms do make strong use of this ability).



\subsection{Model and Notation} \label{sec:model}
We model a communication network as a weighted, undirected, connected $n$-node graph $G = (V, E)$. Every  node has limited initial knowledge. Specifically, assume that each node is associated with a distinct identity number  (e.g., its IP address).
At the beginning of the computation, each node $v$ accepts as input its own identity number and the identity numbers of its neighbors in $G$. The node may also accept some additional inputs as specified by the problem at hand. The nodes are allowed to communicate through the edges of the graph $G$. We assume that the communication occurs in  synchronous  {\em rounds}.
We will use only small-sized messages. In particular, in each round, each node $v$ is allowed to send a message of size $O(\log n)$ through each edge $e = (v, u)$ that is adjacent to $v$.  The message  will arrive at $u$ at the end of the current round.  We also assume that all edge weights are at most polynomial in $n$, and thus in a single round a distance or node ID can be sent through each edge.  A \emph{word} is a block of $O(\log n)$ bits that is sufficient to store either a node ID or a network distance.

This is a  widely used  standard model  to study distributed algorithms (called the {\em CONGEST model}, e.g., see \cite{peleg, PK09}) and captures the bandwidth constraints inherent in real-world computer  networks.  (For example, classical network algorithms
studied include algorithms for shortest paths (Bellman-Ford, Dijsktra), minimum spanning trees etc.) Our algorithms can be easily generalized if $B$ bits  are allowed (for any pre-specified parameter $B$) to be sent through each edge in a round. Typically, as assumed here, $B = O(\log n)$, which is number of bits needed to send a node id in an $n$-node network.
We assume that $n$ (or some constant factor estimate of $n$) is common knowledge among nodes in the network. 

Every edge in the network has some nonnegative weight associated with it, and the distance between two nodes is the minimum, over all paths between the nodes, of the sum of the weights of the edges on the path.  In other words, the normal shortest-path distance with edge weights.  We let $d(u,v)$ denote this distance for all $u,v \in V$.  For a set $A \subseteq V$ and a node $u \in V$, we define the distance from the node to the set to be $d(u,A) = \min\{d(u,a) : a \in A\}$.  For a node $u \in V$ and real number $r \in \mathbb{R}^{\geq 0}$, the \emph{ball} around $u$ of radius $r$ is defined to be $B(u,r) = \{v \in V : d(u,v) \leq r\}$.

The \emph{hop-diameter} $D$ of $G$ is defined to be the maximum over all pairs $u,v$ in $V$ of  the number of hops between $u$ and $v$.  In other words, its the maximum over all pairs of the distance between $u$ and $v$ but where distance is computed assuming that all edge weights are $1$, rather than their actual weights.  The \emph{shortest-path diameter} $S$ of $G$ is slightly more complicated to define.  For $u,v \in V$, let $\mathcal P_{u,v}$ be the set of simple paths between $u$ and $v$ with total weight equal to $d(u,v)$ (by the definition of $d(u,v)$ there is at least one such path).  Let $h(u,v)$ be the minimum, over all paths in $P \in \mathcal P_{u,v}$, of the number of hops in $P$ (i.e.~the number of edges).  Then $S = \max_{u,v \in V} h(u,v)$.  It is easy to see that $D \leq S$  and in general, any method of computing the distance from $u$ to $v$ must use at least $S$ rounds (or else the shortest path will not be discovered).


\section{Distributed Sketches} \label{sec:main}

We are concerned with the problem of constructing a
distance labeling scheme in a distributed manner.  Given an input (weighted) graph $G =
(V,E)$, we want a distributed algorithm so that at termination every
node $u \in V$ knows a small label (or sketch) $L(u)$ with the
property that we can (quickly) compute an approximation to the
distance between $u$ and $v$ just from $L(u)$ and $L(v)$. Since the requirement of
sketch sizes and latency in distance computation may vary from application to application, typically
one would like a trade-off between the distance approximation and these parameters.
%

\subsection{Thorup-Zwick Construction}
The famous algorithm for constructing distance sketches in a centralized manner is Thorup-Zwick~\cite{TZ05}, which works as follows.  They first create a hierarchy of node sets: $A_0 = V$, and for $1 \leq i \leq k-1$, we get $A_i$ by randomly sampling every vertex in $A_{i-1}$ with probability $n^{-1/k}$, i.e.~every vertex in $A_{i-1}$ is included in $A_i$ with probability $n^{-1/k}$.  We set $A_k = \emptyset$ and $d(u, A_k) = \infty$ by definition.


Let $p_i(u)$ be the vertex in $A_i$ with minimum distance from $u$.  Let $B_i(u) = \{w \in A_i : d(u,w) < d(u,A_{i+1})\}$, and let $B(u) = \cup_{i=0}^{k-1} B_i(u)$ (for now we will assume that all distances are distinct; this can be made without loss of generality by breaking ties consistently through processor IDs or some other method).  $B(u)$ is called the \emph{bunch} of $u$.  The label $L(u)$ of $u$ consists of all nodes $\{p_i(u)\}_{i=0}^{k-1}$ and $B(u)$, as well as the distances to all of these nodes.  Thorup and Zwick showed that these labels are enough to approximately compute the distance, and that these labels are small.  We give sketches of these proofs for completeness.

\begin{lemma}[\cite{TZ05}] \label{lem:TZ_small}
For all $u \in V$, the expected size of $L(u)$ is at most $O(k n^{1/k})$ words.
\end{lemma}
\begin{proof}
We prove that the expected size of $B_i(u)$ is at most $n^{1/k}$ for every $0 \leq i \leq k-1$, which clearly implies the lemma via linearity of expectation.  Suppose that we have already made the random decisions that define levels $A_0, \dots, A_i$, and now for each $v \in A_i$ we flip the coin to see if it is also in $A_{i+1}$.  If we flip these coins in order of distance from $u$ (this is just in the analysis; the algorithm can flip the coins simultaneously or in arbitrary order) then the size of $B_i(u)$ is just the number of coins we flip before we see a heads, where the probability of flipping a heads is $n^{-1/k}$.  In expectation this is $n^{1/k}$.
\end{proof}

\begin{lemma}[\cite{TZ05}] \label{lem:TZ_stretch}
Given $L(u)$ and $L(v)$ for some $u,v \in V$, we can compute a distance estimate $d'(u,v)$ with $d(u,v) \leq d'(u,v) \leq (2k-1) d(u,v)$ in time $O(k)$.
\end{lemma}
\begin{proof}
For each $0 \leq i \leq k-1$, we check whether $p_i(u) \in B_i(v)$ or $p_i(v) \in B_i(u)$.  Let $i^*$ be the first level at which at least one of these events occurs.  Note that $i^*$ is well-defined and is at most $k-1$, since by definition $p_{k-1}(u) \in B_{k-1}(v)$ and $p_{k-1}(v) \in B_{k-1}(u)$.  If the first condition is true then we return distance estimate $d'(u,v) = d(u, p_{i^*}(u)) + d(v, p_{i^*}(u))$, and if the second condition is true then we return $d'(u,v) = d(u, p_{i^*}(v)) + d(v, p_{i^*}(v))$.  Note that the necessary distances are in the labels as part of $B_{i^*}(u)$ and $B_{i^*}(v)$, so we can indeed compute this from $L(u)$ and $L(v)$.

We first prove by induction that $d(u, p_i(u)) \leq i \cdot d(u,v)$ and $d(v, p_i(v)) \leq i \cdot d(u,v)$ for all $i \leq i^*$.  In the base case, when $i=0$, both inequalities are true by definition.  For the inductive step, let $1 \leq i \leq i^*$.  Since $i \leq i^*$ we know that $i-1 < i^*$, so $p_{i-1}(u) \not\in B_{i-1}(v)$ and $p_{i-1}(v) \not \in B_{i-1}(u)$.  This implies that $d(v, p_{i}(v)) \leq d(v, p_{i-1}(u)) \leq d(v,u) + d(u, p_{i-1}(u) \leq d(u,v) + (i-1) d(u,v) = i \cdot d(u,v)$, where the first inequality is from $i-1 < i^*$, the second is from the triangle inequality, and the third is from the inductive hypothesis.  Similarly, we get that  and $d(u, p_i(u)) \leq i \cdot d(u,v)$.

Now suppose without loss of generality that $p_{i^*}(v) \in B_{i^*}(u)$ (if the roles are reversed we can just switch the names of $u$ and $v$).  Then our distance estimate is $d'(u,v) = d(u, p_{i^*}(v)) + d(v, p_{i^*}(v)) \leq d(u,v) + 2d(v, p_{i^*}(v)) \leq d(u,v) + 2 i^* d(u,v) = (2i^* + 1) d(u,v)$, where the first inequality is from the triangle inequality and the second is from our previous inductive proof.  Since $i^* \leq k-1$, this gives a stretch bound of $2k-1$ as claimed.
\end{proof}

\subsection{Distributed Algorithm}
The natural question is whether we can construct these labels in a
distributed manner.  For a vertex $v \in A_i \setminus A_{i+1}$, let
$C(v) = \{w \in V : d(w,v) < d(w, A_{i+1})$.  This is called the
\emph{cluster} of $v$.  Note that the clusters are the inverse of the
bunches: $u \in C(v)$ if and only if $v \in B(u)$.  So we will
construct a distributed algorithm in which every vertex $u$ knows
exactly which clusters it is in and its distance from the centers of
those clusters, and thus is able to construct its label.  Also, it's easy to
see that the clusters are connected: if $u \in C(v)$ then obviously any
vertex $w$ on the shortest path from $u$ to $v$ is also in $C(v)$.

The distributed protocol is as follows: we first divide into $k$ phases,
where in phase $i$ we deal with clusters from vertices in $A_i
\setminus A_{i+1}$.  However, we run the phases from top to bottom --
we first do phase $k-1$, then phase $k-2$, down to phase $0$.  In
phase $i$ the goal is for every node $u \in V$ to know for every node
$v \in V$ whether $u \in C(v)$, and if so, its distance to $v$.  Thus every node $u$ will know $B_i(u)$ at the end of phase $i$.  We will give an upper bound on the length of each phase with respect to $n$ and $S$, so if every node knows both $n$ and $S$ they can all start each phase together by waiting until the upper bound is met.     For now we will make the assumption that every node knows $S$ (the shortest path diameter), thus solving the issue of synchronizing the beginning of the phases, but we will show in Section~\ref{sec:termination} how to remove this assumption.

Let us first consider phase $k-1$, i.e.~the first phase that is run.  This phase is especially simple, since $B_{k-1}(u) = A_{k-1}$ for every node $u$.  So at the end of this phase we simply want every node to know all of the nodes in $A_{k-1}$ and its distances to all of them.  This is known as the \emph{$k$-Source Shortest Paths Problem}, and can be done in $O(|A_{k-1}| S)$ rounds using $O(|A_{k-1}| |E| S)$ messages by running distributed Bellman-Ford from each node in $A_{k-1}$ simultaneously~\cite{PK09}.  In particular, for a fixed source $v \in A_{k-1}$ every node $u \in V$ runs the following protocol: initially, $u$ guesses that its distance to $v$ is $d'(u,v) = \infty$.  If it hears a message from a neighbor $w$ that contains a distance $a(w)$, then it checks if $d(u,w) + a(w) < d'(u,v)$.  If so, then it updates $d'(u,v)$ to $d(u,w) + a(w)$ and sends to all its neighbors a message that contains the new $d'(u,v)$.  This algorithm is given in detail as Algorithm~\ref{alg:BF}.

\begin{algorithm}[]
\caption{Basic Bellman-Ford for $u$}
\label{alg:BF}
\SetKwInOut{init}{Initialization}
\init{$d' = \infty$}
For each neighbor $w$ of $u$, get message $a(w)$
\\
$z \leftarrow \min_{w \in N(u)} \{a(w) + d(u,w)\}$
\\
\If{$z < d'$}{
$d' \leftarrow z$
\\
Send message $d'$ to all neighbors}
\end{algorithm}

In this description we assume that there is one source $v$ that is known to all nodes, but this clearly is not necessary.  With multiple unknown sources each message could also contain the ID of the source and each node $u$ could keep track of its guesses $d'(u, \cdot)$ for every source that it has seen at least one message from.  The standard analysis of Bellman-Ford (see e.g.~\cite{PK09}) gives the following lemmas:
\begin{lemma} \label{lem:correct_first_phase}
At the end of phase $k-1$, every node $u \in V$ knows which vertices are in $A_{k-1}$ as well as $d(u,v)$ for all $v \in A_{k-1}$.
\end{lemma}
\begin{lemma} \label{lem:complexity_first_phase}
Phase $k-1$ completes after at most $O(|A_{k-1}| S)$ rounds and at most $O(|E| \cdot |A_{k-1}| S)$ messages.
\end{lemma}

To handle phase $i$, we will assume inductively that $B_{i+1}(u)$ is known to $u$ at the start of phase $i$, as well as the distance from $u$ to every node in $B_{i+1}(u)$.  In particular, we will assume that $u$ knows its distance to the \emph{closest} node in $A_{i+1}$, i.e.~$d(u, A_{i+1})$.  In phase $i$ we will simply use a modified version of Bellman-Ford in which the sources are $A_{i} \setminus A_{i+1}$, but node $u$ only ``participates" in the algorithm for sources $v \in A_i \setminus A_{i+1}$ when it gets a message that implies that $d(u,v) < d(u,A_{i+1})$, i.e.~that $v \in B_i(u)$.  To handle the multiple sources, each node $u$ will maintain for every possible source $v \in V$ an outgoing message queue, which will only ever have a $0$ or $1$ message in it.  $u$ just does round-robin scheduling among the nonempty queues, sending the current message to all neighbors and removing it from the queue.  To simplify the code, we will assume without loss of generality that $V = \{0,1,\dots,n-1\}$ (this assumption is only used to simplify the round-robin scheduler, and can easily be removed).  This algorithm is given as Algorithm~\ref{alg:BF_modified}.

\begin{algorithm}[]
\caption{Modified Bellman-Ford for node $u$ in phase $i$}
\label{alg:BF_modified}
\SetKwInOut{init}{Initialization}
Initialization:
\\
\ForEach{$v \in V\setminus \{u\}$}{
$d'(v) \leftarrow \infty$
\\
$q(v) \leftarrow 0$
\\
$i \leftarrow 0$
}
\BlankLine
In the first round:
\\
\If{$u \in A_i \setminus A_{i+1}$}{Send message $\langle u, 0 \rangle$ to all neighbors}
\BlankLine
In each round:
\\
\tcp{Receive and process new messages}
\ForEach{$w \in N(u)$}{
Get message $m(w) = \langle v_w, a_w \rangle$
\\
\If{$a_w + d(u,w) < d(u, A_{i+1}) \land a_w + d(u,w) < d'(v_w)$}{
$d'(v_w) \leftarrow a_w + d(u,w)$
\\
$q(v_w) \leftarrow 1$
}}
\tcp{Send message from next nonempty queue}
$i' \leftarrow i$
\\
$i \leftarrow (i+1) \% n$
\\
\lWhile{$q(i) == 0 \land i \neq i'$}{$i \leftarrow (i+1) \% n$}
\\
\If{q(i) == 1}{
Send message $\langle i, d'(i) \rangle$ to all neighbors
\\
$q(i) \leftarrow 0$
}
\end{algorithm}

\begin{lemma} \label{lem:correct}
At the end of phase $i$, every node $u \in V$ knows $B_i(u)$ and its
distance to all nodes in $B_i(u)$.
\end{lemma}
\begin{proof}
We prove this by induction on the phase.  The base case is phase
$k-1$, which is satisfied by Lemma~\ref{lem:correct_first_phase}.  Now consider some phase $i \geq 0$.  Let $v \in
A_{i} \setminus A_{i+1}$ -- we will show by induction on the hop count
of the shortest path that all nodes $u \in C(v)$ find out their
distance to $v$.  If $u \in C(v)$ is adjacent to $v$ via a shortest path, then
obviously after the first round it will know its distance to $v$.  If
$u \in C(v)$ is not adjacent to $v$ via a shortest path, then by
induction the next hop on the shortest path from $u$ to $v$ finds out
its correct distance to $v$, and thus will forward the announcement to
$u$.  Thus at the end of phase $i$ every node $u \in C(v)$ knows its
distance from $v$, and since this holds for every $v \in A_{i}
\setminus A_{i+1}$ we have that every $u \in V$ knows its distance to
all nodes in $B_i(u)$.
\end{proof}

Before we prove the time and message complexity bounds, we give a lemma that extends the expected size analysis of~\cite{TZ05} to give explicit tail bounds on the probability that the construction is large:

\begin{lemma} \label{lem:bunch_size}
For every $i \in \{0,\dots, k-1\}$ and every $u \in V$, the probability that $|B_i(u)| > O(n^{1/k} \ln n)$ is at most $1/n^3$.
\end{lemma}
\begin{proof}
In order for $|B_i(u)| > 3 n^{1/k} \ln n$, the closest $3 n^{1/k} \ln n$ nodes in $A_i$ to $u$ must all decide not to be part of $A_{i+1}$.  Since the probability that any particular node in $A_i$ joins $A_{i+1}$ is $n^{-1/k}$, the probability that this happens is at most $(1-n^{-1/k})^{3 n^{1/k} \ln n} \leq e^{-3 \ln n} = 1/n^3$.
\end{proof}

We can now bound the time and message complexity of each phase:

\begin{lemma} \label{lem:complexity}
With high probability, each phase takes $O(n^{1/k} S \log n)$ rounds and $O(n^{1/k} S |E| \log n)$ messages.
\end{lemma}
\begin{proof}
  Let $v \in A_{i} \setminus A_{i+1}$.  Intuitively, if $v$ were the only vertex in
  $A_i \setminus A_{i+1}$, then in phase $i$ the algorithm devolves
  into distributed Bellman-Ford in which the only vertices that ever forward messages are
  vertices in $C(v)$.  This would clearly take $O(S)$ rounds.  In the
  general case, each vertex $u$ only participates in $O(|B_i(u)|) =
  O(n^{1/k} \log n)$ of these shortest path algorithms, so each ``round'' of
  the original algorithm can be split up into $O(n^{1/k} \log n)$ rounds to
  accommodate all of the different sources.  Thus the total time taken
  is $O(n^{1/k} S \log n)$ as claimed.

  To prove this formally, let $u \in V$ and $v \in A_i \setminus A_{i+1}$, with $v \in B_i(u)$.  Let $v = v_0, v_1, \dots, v_{\ell-1} = u$ be a shortest path from $v$ to $u$ with the fewest number of hops.  Thus $\ell \leq S$.  We prove by induction that $v_j$ receives a message $\langle v, d(v_{j-1}, v) \rangle$ at time at most $O(n^{1/k} j \log n)$, which clearly implies the lemma.  For the base case, in the first round $v$ sends out the message $\langle v, 0 \rangle$ to its neighbors, so $v_1$ receives the correct message at time $1 \leq n^{1/k} \log n$.  For the inductive step, consider node $v_j$.  We know by induction that $v_{j-1}$ received a message $\langle v, d(v_{j-2}, v) \rangle$ at time at most $t = O(n^{1/k} (j-1) \log n)$.  If $v_{j-1}$ already knew this distance from $v$, then it also already sent a message (or put one in the queue) informing its neighbors ($v_j$ in particular) about this.  Otherwise, $v_{j-1}$ puts a message in its outgoing queue at time $t$.  Since the nonempty queues are processed in a round-robin manner, and by Lemma~\ref{lem:bunch_size} at most $O(n^{1/k} \log n)$ queues are ever nonempty throughout the phase, $v_{j-1}$ sends a message $\langle v, d(v_{j-1} v) \rangle$ at time at most $t + O(n^{1/k} \log n) = O(n^{1/k} j \log n)$, as claimed.

  The time complexity bound immediately implies the message complexity bound, since in every round there are at most $2$ messages on each edge (one in each direction).
  \end{proof}

Lemmas~\ref{lem:TZ_small}, \ref{lem:correct}, and \ref{lem:complexity} obviously imply the following theorem:

\begin{theorem} \label{thm:main}
For any $k \geq 1$, there is a distributed sketching algorithm that takes $O(k n^{1/k} S \log n)$ rounds  and $O(k n^{1/k} S |E| \log n)$ messages, after which with high probability every node has a sketch of size at most $O(k n^{1/k} \log n)$ words (and expected size $O(k n^{1/k})$ words) that provides approximate distances with stretch $2k-1$.
\end{theorem}

\subsection{Termination Detection} \label{sec:termination}
We now show how to remove the assumption that every node knows $S$.  Note that we do not show how to satisfy this assumption, i.e.~we do not give an algorithm that computes $S$ and distributes it to all nodes.  Rather, we show how to detect when a phase has terminated, and thus when a new phase should start.  We use basically the same termination detection algorithm as the one used by Khan et al.~\cite{khan-podc}, just adapted to our context.

At the very beginning of the algorithm, even before phase $k-1$, we run a leader election algorithm to designate some arbitrary vertex $r$ as the \emph{leader}, and then build a breadth-first search (BFS) tree $T$ out of $r$ so that every node knows its parent in the tree as well as its children.  This can be done in $O(D) \leq O(S)$ rounds and $O(|E| \log n)$ messages~\cite{khan-podc}.


At the beginning of phase $i$, the leader $r$ sends a message to all nodes (along $T$) telling them when they should start phase $i$, so they all begin together.  We say that a node $u$ is \emph{complete} if either $u \not \in A_i \setminus A_{i+1}$ or every vertex in $C(u)$ knows its distance to $u$ (we will see later how to use echo messages to know when this is the case).  So initially the only complete nodes are the ones not in $A_i \setminus A_{i+1}$.  Any such node that is also a leaf in $T$ immediately sends a COMPLETE message to its parent in the tree.  Throughout the phase, when any node has heard COMPLETE messages from all of its children in $T$ and is itself complete, it sends a COMPLETE message to its parent.

Now suppose that when running phase $i$, some node $u$ gets a message $m(w) = \langle v_w, a_w \rangle$ from a neighbor $w$.  There are two reasons that this message might not result in a new message added to the send queue: if $a_w + d(u,w) \geq d(u, A_{i+1})$ ($v_w$ has not yet been shown to be in $B_i(u)$), or if $a_w + d(u,w) \geq d'(v_w)$ ($u$ already knows a shorter path to $v_w$).  Furthermore, even if $m(w)$ does result in a new message added to the send queue, it might get superseded by a new message added to the queue with an updated value of $d'(v_w)$ before the value based on $m(w)$ can be sent.  All three of these conditions can be tracked by $u$, so for each message $m$ that $u$ receives (say from neighbor $w$) it keeps track of whether or not it sends out a new message based on $m$.  If it does not (one of the two conditions failed, or it was superseded), then it sends an ECHO message back to $w$, together with a copy of the message.  If $u$ does send out a new message based on $m$, then when it has received ECHO messages for $m$ from all of its neighbors (except for $w$) it sends an ECHO message to $w$ together with a copy of $m$.

It is easy to see inductively that when a node $u$ sends a message $m$, it will also know via ECHO messages when $m$ has ceased to propagate in the network since all of its neighbors will have ECHO'd it back to $u$.  So if $u \in A_i \setminus A_{i+1}$, and thus only sends out one message that has first coordinate $u$, it will know when this message has stopped propagating, which clearly implies that every node $v \in V$ that is in $C(u)$ knows its correct distance to $u$ (as well as the fact that $u \in B_i(v)$).  At this point $u$ is complete, so once it has received COMPLETE messages from all of its children in $T$ it will send a COMPLETE message to its parent.

Once $r$ has received COMPLETE messages from all of its children (and is itself complete) the phase is over.  So $r$ starts the next phase by sending a START message to all nodes using the $T$, and the next phase begins.

It is easy to see that the ECHOs only double the number of messages and rounds, since any message sent along an edge corresponds to exactly one ECHO sent back the other way.  Electing a leader and building a BFS tree take only a negligible number of messages and rounds compared to the bounds of Theorem~\ref{thm:main}.  Each node sends only one COMPLETE message, so there are at most $O(n)$ COMPLETE messages which is tiny compared to the bound in Theorem~\ref{thm:main}, and the number of extra rounds due to COMPLETE messages is clearly only $O(D)$.  Thus even with the extra termination detection, the bounds of Theorem~\ref{thm:main} still hold.

\section{Sketches with Slack} \label{sec:slack}

Let $u,v \in V$.  We say that $v$ is \emph{$\epsilon$-far} from $u$ if
$|\{w : d(u,w) < d(u,v)\}| \geq \epsilon n$, i.e.~if $v$ is not one of
the $\epsilon n$ closest nodes to $u$.  Given a labeling $L(u)$ for
each $u \in V$, we say that it has stretch $2k-1$ and
$\epsilon$-slack if the distance that we compute for $u$ and
$v$ given $L(u)$ and $L(v)$ is at least $d(u,v)$ and at most $(2k-1)
d(u,v)$ for all $u,v \in V$ where $v$ is $\epsilon$-far from $u$.
Labelings with slack were previously studied by Chan, Dinitz, and
Gupta~\cite{CDG06} and Abraham, Bartal, Chan, Dhamdhere, Gupta, Kleinberg, Neiman, and
Slivkins~\cite{ABCDGKNS05}.  The main technique of Chan et al.~was the use of a new type
of net they called a \emph{density net}.  For each $u \in V$, let
$R(u,\epsilon) = \inf \{r : |B(u,r)| \geq \epsilon n\}$ be the minimum
distance necessary for the ball around $u$ to contain at least
$\epsilon n$ points, and let $B^{\epsilon}(u) = B(u, R(u,\epsilon))$ be this ball.  We give a definition of density net that is slightly modified from~\cite{CDG06} in order to make it easier to work with in a distributed context.

\begin{definition}
A set of vertices $N \subseteq V$ is an $\epsilon$-density net if:
\begin{enumerate}
\item For all $u \in V$, there is a vertex $v \in N$ such that $d(u,v)
  \leq R(u,\epsilon)$, and
\item $|N| \leq \frac{10}{\epsilon} \ln n$.
\end{enumerate}
\end{definition}

Chan et al.\ give a centralized algorithm that computes an
$\epsilon$-density net in polynomial time for any $\epsilon$.  Their density nets are somewhat different, in that they contain only $1/\epsilon$ nodes but the closest net node to $u$ is only guaranteed to be within $2 R(u,\epsilon)$ instead of $R(u, \epsilon)$.  We modify these values in order to give a distributed construction, and in fact with these modifications it is trivial to build density nets via random sampling.

\begin{lemma}\label{lem:density_net}
  There is a distributed algorithm that, with high probability,
  constructs an $\epsilon$-density net in constant time.
\end{lemma}
\begin{proof}
The algorithm is simple: every vertex  independently chooses to be in
$N$ with probability $\frac{5\ln n}{\epsilon n}$.  The expected size of $N$
is clearly $\frac{5\ln n}{\epsilon}$, and by a simple Chernoff bound (see e.g.~\cite{MR95}) we
have that the probability that  $|N| > \frac{10\ln n}{\epsilon}$ is at most $e^{-(20\ln n)/(3\epsilon)} \leq 1/n^{6/\epsilon}$,
so the second constraint is satisfied with high probability.

For the first constraint, that for every vertex $u$ there is some vertex $v \in B^{\epsilon}(u) \cap N$, we split into two cases depending on $\epsilon$.  If $\epsilon \leq \frac{5\ln n}{n}$, then every node has probability $1$ of being in $N$, so the condition is trivially satisfied.  Otherwise we have $\epsilon > \frac{5\ln n}{n}$, so for every $u$ we have $|B^{\epsilon}(u)| \geq 5\ln n$ and the expected size of $B^{\epsilon}(u) \cap N$ is exactly $5\ln n$.  Using a similar Chernoff bound (but from the other direction) gives us that the probability that $|B^{\epsilon}(u) \cap N|$ is less than $1$ is at most $e^{-(25 \ln n)/8} \leq 1/n^3$.  Now we can just take a union bound over all $u$ to get that the first constraint is satisfied with high probability.
\end{proof}

Using this construction, we can efficiently construct short sketches with $\epsilon$-slack:

\begin{theorem}
There is a distributed algorithm that uses at most $O(S\frac{1}{\epsilon} \log n)$ rounds and at most $O(S |E| \frac{1}{\epsilon} \log n)$ messages so that at the end of the algorithm, every node has a sketch of size at most $O(\frac{1}{\epsilon} \log n)$ words with stretch at most $3$ and $\epsilon$-slack.
\end{theorem}
\begin{proof}
The algorithm first uses Lemma~\ref{lem:density_net} to construct an $\epsilon$-density net $N$.  It is easy to see that the closest net
points to $u$ and $v$ are a good approximation to the distance between
$u$ and $v$.  In particular, suppose that $v$ is $\epsilon$-far from
$u$, let $u'$ be the closest node in $N$ to $u$, and let $v'$ be
the closest node in $N$ to $v$.  Then $d(u, u') \leq R(u,\epsilon)
\leq d(u,v)$ by the definition of $\epsilon$-far and
$\epsilon$-density nets, $d(v, u') \leq d(v,u) +
d(u, u') \leq 2d(u,v)$, and thus $d(u, u') + d(v, u') \leq 3 d(u,v)$.  This means that if every vertex keeps as its sketch its distance from
all nodes in $N$, we will have sketches with $\epsilon$-uniform slack
and stretch $3$ (to compute an approximation to $d(u,v)$ we can just consider use $\min_{w \in N} \{d(u,w) + d(w,v)\}$, which we can compute from the two sketches).  The size of these sketches is clearly $O(|N|) = O(\frac{1}{\epsilon} \log n)$, since for every node in $N$ we just need to store its ID and its distance.

It just remains to show how to compute these sketches efficiently.  But this is simple, since it is exactly the $k$-Source Shortest Paths problem where the sources are the nodes in $N$.  So we simply run the $k$-source version of Distributed Bellman-Ford, which gives the claimed time and message complexity bounds.
\end{proof}

We can get a different tradeoff by applying Thorup and Zwick to the
density net itself, instead of simply having every node remember its distance to all net nodes.  This is essentially what is done in the slack labeling schemes of Chan et al.~\cite{CDG06}, just with slightly different parameters and constructions (since they were able to use centralized constructions).  In particular, suppose that we manage to use Thorup-Zwick on the net, so the distances between net points are preserved up to stretch $2k-1$.  Then these sketches would have size at most $O(k |N|^{1/k} \log n) = O(k \left(\frac{1}{\epsilon} \log n\right)^{1/k} \log n)$.  For each $u \in V$, let $u' \in N$ be the closest node in the density net to $u$.  We let the sketch of $u$ be the identity of $u'$, the distance between $u$ and $u'$, and the Thorup-Zwick label of $u'$.  We call this the \emph{$(\epsilon, k)$-CDG sketch}.  Clearly this sketch has size $O(k \left(\frac{1}{\epsilon} \log n\right)^{1/k} \log n)$.  Let $u,v \in V$ such that $v$ is $\epsilon$-far from $u$.  Our estimate of the distance will be $d(u, u') + d''(u', v') + d(v', v)$, where $d''(u', u') \leq (2k-1) d(u', v')$ is the approximate distance given by the Thorup-Zwick labels.  This can obviously be computed given the sketches for $u$ and $v$.  To bound the stretch, we simply use the definition of density nets, the triangle inequality, and the fact that $v$ is $\epsilon$-far from $u$. This gives us a distance estimate $d'(u,v)$ with

\begin{align*}
d'(u,v) &= d(u, u') + d''(u', v') + d(v', v) \\
&\leq d(u,u') + (2k-1) d(u', v') + d(v',v) \\
&\leq d(u,v) + (2k-1) (d(u', u) + d(u,v) + d(v, v')) \\
& \ \ \ \ \ + 2d(u,v) \\
&\leq 3d(u,v) + (2k-1) (4d(u,v)) \\
&= (8k-1) d(u,v)
\end{align*}

This gives the basic lemma about these sketches, which was proved by~\cite{CDG06} (modulo our modifications to density nets):

\begin{lemma}[\cite{CDG06}] \label{lem:CDG_basic}
For any $\epsilon > 0$ and $1 \leq k \leq O(\log \frac{1}{\epsilon})$, with high probability the $(\epsilon, k)$-CDG sketch has size at most $O(k \left(\frac{1}{\epsilon} \log n\right)^{1/k} \log n)$ words and $(8k-1)$-stretch with $\epsilon$-slack.
\end{lemma}

It remains to show how to construct $(\epsilon, k)$-CDG sketches in a distributed manner, which boils down to modifying the algorithm of Theorem~\ref{thm:main} to work with the density net rather than with the full point set (note that this is trivial in a centralized setting since we can just consider the metric completion).

\begin{lemma} \label{lem:CDG_alg}
For any $\epsilon > 0$ and $1 \leq k \leq O(\frac{1}{\epsilon})$, there is a distributed algorithm so that after $O\left(k S \left(\frac{1}{\epsilon} \log n\right)^{1/k} \log n\right)$ rounds and $O\left(k S |E|\left(\frac{1}{\epsilon} \log n\right)^{1/k} \log n\right)$ messages every node knows its $(\epsilon, k)$-CDG sketch
\end{lemma}
\begin{proof}
We first apply Lemma~\ref{lem:density_net} to construct the $\epsilon$-density net $N$.  We now want every node $u$ to know its closest net node $u'$ and its distance from $u'$.  This can be done via a single use of Distributed Bellman-Ford, where we just imagine a ``super node" consisting of all of $N$.  This takes $O(S)$ rounds and $O(S |E|)$ messages.

Now we need to run Thorup-Zwick on $N$.  But this is easy to do, since we just modify the $A_i$ sets to be subsets of $N$ instead of $V$ and change the sampling probability from $n^{-1/k}$ to $\left(\frac{10}{\epsilon} \ln n\right)^{-1/k}$.  Note that for every node $u \in V$ the bunch $B_i(u)$ is still well defined, and with high probability has size at most $O\left(\left(\frac{1}{\epsilon} \log n\right)^{1/k} \log n\right)$ (via an argument analogous to Lemma~\ref{lem:bunch_size}).  This means that we can run Algorithm~\ref{alg:BF_modified} using these new $A_i$ sets and every node will know their Thorup-Zwick sketch for these $A_i$ sets.  In particular, the nodes in $N$ will have a sketch that is exactly equal to the sketch they would have if we ran Algorithm~\ref{alg:BF_modified} on the metric completion of $N$, rather than on $G$.  It is easy to see that Lemma~\ref{lem:complexity} still applies but with $n^{1/k} \log n$ (the upper bound on the size of each $B_i(u)$) changed to $O\left(\left(\frac{1}{\epsilon} \log n\right)^{1/k} \log n\right)$, so each phase takes $O\left(S \left(\frac{1}{\epsilon} \log n\right)^{1/k} \log n\right)$ rounds and $O\left(S |E|\left(\frac{1}{\epsilon} \log n\right)^{1/k} \log n\right)$ messages.  Since there are $k$ phases, this gives the desired complexity bounds.

As before, this assumes that every node knows $S$ in order to synchronize the phases.  However, we can remove this assumption by using the termination detection algorithm of Section~\ref{sec:termination}.  This at most doubles the number of messages and rounds and adds an extra $O(|E| \log n)$ messages and $O(D)$ rounds, which is negligible.
\end{proof}

Combining Lemmas~\ref{lem:CDG_basic} and~\ref{lem:CDG_alg} gives us the following theorem.

\begin{theorem} \label{thm:slack_main}
For any $\epsilon > 0$ and $1 \leq k \leq O(\log \frac{1}{\epsilon})$, there is a distributed sketching algorithm that completes in at most $O\left(k S \left(\frac{1}{\epsilon} \log n\right)^{1/k} \log n\right)$ rounds and $O\left(k S |E|\left(\frac{1}{\epsilon} \log n\right)^{1/k} \log n\right)$ messages, after which with high probability every node has a sketch of size at most  $O(k \left(\frac{1}{\epsilon} \log n\right)^{1/k} \log n)$ words that provides approximate distances with stretch $8k-1$ and $\epsilon$-slack.
\end{theorem}

\subsection{Gracefully Degrading Sketches and Average Stretch} \label{sec:average}
We now show how to use Theorem~\ref{thm:slack_main} to construct sketches with bounded \emph{average} stretch, as well as bounded worst-case stretch.  Formally, suppose that we have a weighted graph $G= (V,E)$ that induces the metric $d$ and a sketching algorithm that allows us to compute distance estimates $d'$ with the property that $d'(u,v) \geq d(u,v)$ for all $u, v \in V$.  The \emph{average stretch} of the sketching algorithm is $\frac{1}{{n \choose 2}} \sum_{\{u,v\} \in {V \choose 2}} \frac{d'(u,v)}{d(u,v)}$.

In fact, we will prove a stronger statement, that there are good distributed algorithms for computing \emph{gracefully degrading} sketches.  A sketching algorithm is gracefully degrading with $f(\epsilon)$ stretch if for every $\epsilon \in (0,1)$ it is a sketch with stretch $f(\epsilon)$ and $\epsilon$-slack.  In other words, instead of specifying $\epsilon$ ahead of time (as in the slack constructions) we need a single sketch that works simultaneously for every $\epsilon$.  It is easy to see that when $f$ is $O(\log \ef)$, gracefully degrading sketches provide the desired average and worst-case stretch bounds (this was implicit in Chan et al.~\cite{CDG06}, but they only formally showed this for their specific gracefully-degrading construction, which is slightly different than ours):

\begin{lemma} \label{lem:GD_average}
Any gracefully degrading sketching algorithm with $O(\log \ef)$ stretch has stretch at most $O(\log n)$ and average stretch at most $O(1)$.
\end{lemma}
\begin{proof}
The bound on the worst-case stretch is immediate by setting $\epsilon < \frac{1}{n}$.  With this setting of $\epsilon$, every two points are $\epsilon$-far from each other, and thus the stretch bound of $O(\log \ef) = O(\log n)$ holds for all pairs.

To bound the average stretch, for each $1 \leq i \leq \log n$ and vertex $u \in V$ let $A(u,i) = B^{1/2^{i-1}}(u) \cap (V \setminus B^{1/2^i}(u))$.  In other words, $A(u,i)$ is the set of points that are outside the smallest ball around $u$ containing at least $n/2^i$ points, but inside the smallest ball around $u$ containing at least $n/2^{i-1}$ points.  Note that $|A(u,i)| = n/2^i$.  Furthermore, we can bound the stretch between $u$ and any node in $A(u,i)$ by $O(i)$, since when we set $\epsilon = 1/2^i$ we have a stretch bound of $O(\log \ef) = O(i)$ for the nodes in $A(u,i)$.  Then the average stretch is at most
\begin{align*}
\frac{1}{{n \choose 2}} \sum_{\{u,v\} \in {V \choose 2}} \frac{d'(u,v)}{d(u,v)} &\leq \frac{1}{n(n-1)} \sum_{u \in V} \sum_{v \neq u} \frac{d'(u,v)}{d(u,v)} \\
&\leq \frac{1}{n(n-1)} \sum_{u \in V} \sum_{i=1}^{\log n} \sum_{v \in A(u,i)}\frac{d'(u,v)}{d(u,v)} \\
& \leq \frac{1}{n(n-1)} \sum_{u \in V} \sum_{i=1}^{\log n} O(i\cdot \frac{n}{2^i}) \\
& \leq \frac{1}{n(n-1)} \sum_{u \in V} O(n)  \\
& \leq O(1),
\end{align*}
proving the lemma.
\end{proof}

This lemma reduces the problem of constructing sketches with good average stretch to the problem of constructing gracefully degrading sketches.  But this turns out to be simple, given Theorem~\ref{thm:slack_main}.  The intuition behind gracefully degrading sketches is that they work simultaneously for every slack parameter $\epsilon$, so to create them we simply use $O(\log n)$ different sketches with slack, one for each power of $2$ between $1/n$ and $1$.

\begin{theorem} \label{thm:GD_main}
There is a distributed gracefully degrading sketching algorithm that gives sketches of size at most $O(\log^4 n)$ words with $O(\log \ef)$-stretch that completes in at most $O(S \log^4 n)$ rounds and at most $O(S |E| \log^4 n)$ messages.
\end{theorem}
\begin{proof}
Our construction is simple: for every $1 \leq i \leq \log n$ we use Theorem~\ref{thm:slack_main} with slack $\epsilon_i = \frac{1}{2^i}$ and stretch $k = O(\log \frac{1}{\epsilon_i}) = O(\log 2^i)$.  The sketch remembered by a node is just the union of these $O(\log n)$ sketches.  Given the sketches for two different vertices $u$ and $v$ where $v$ is $\epsilon$-far from $u$, we can compute the $O(\log n)$ different distance estimates and take the minimum of them as our estimate.

To see that this is gracefully degrading with stretch $O(\log \ef)$, first note that all of the $O(\log n)$ estimates are at least as large as $d(u,v)$, so we just need to show that at least one of the estimates is at most $O(\log \ef) d(u,v)$.  Let $\epsilon_i$ be $\epsilon$ rounded down to the nearest power of $1/2$.  Then $v$ is obviously $\epsilon_i$-far from $u$, so the estimate for the $\epsilon_i$-sketch will provide an estimate of at most $O(\log \frac{1}{\epsilon_i}) d(u,v) = O(\log \ef) d(u,v)$.

 Theorem~\ref{thm:slack_main}, when specialized to the case of $k = O(\log \ef)$, completes in at most $O(S \log \ef \log^2 n)$ rounds and $O(S |E| \log \ef \log^2 n)$ messages and gives sketches of size $O(\log \ef \log^2 n)$.  Since we just run each of the $O(\log n)$ instantiations of the theorem back to back, the total number of rounds is at most $O(S \log^2 n) \sum_{i=1}^{\log n} \log 2^i = O(S \log^4 n)$, the number of messages is at most $O(S |E| \log^4 n)$, and the size is at most $O(\log^4 n)$.  Note that we can handle determination detection for each of these as usual, based on Section~\ref{sec:termination}.
\end{proof}

Together with Lemma~\ref{lem:GD_average}, this gives the following corollary:

\begin{corollary}
There is a distributed sketching algorithm that give sketches of size at most $O(\log^4 n)$ with $O(\log n)$-stretch and $O(1)$ average stretch that completes in at most $O(S \log^4 n)$ rounds and at most $O(S |E| \log^4 n)$ messages.
\end{corollary}

Note that, when compared to our sketch from Theorem~\ref{thm:main} with $O(\log n)$ stretch, we pay only an extra $O(\log^2 n)$ factor in the size of the sketch as well as the number of rounds and messages, and in return we are able to achieve constant average stretch.

\section{Conclusions}
In this paper we initiated the study from a theoretical point of view of distributed algorithms for computing distance sketches in a network.  We showed that the Thorup-Zwick distance sketches~\cite{TZ05}, which provide an almost-optimal tradeoff between the size of the sketches and their accuracy, can be computed efficiently in a distributed setting, where our notion of efficiency is the standard definition of the number of rounds in the CONGEST model.  Combining this distributed algorithm with centralized techniques of Chan et al.~\cite{CDG06}, that we were also able to turn into efficient distributed algorithms, yielded a combined construction with the same worst-case stretch as the smallest version of Thorup-Zwick, but much better average stretch. This required only a polylogarithmic cost in the size of the sketches and the time necessary to construct them.  These results are a first step towards making the theoretical work on distance sketches more practical, by moving from a centralized setting to a distributed setting.  It would be interesting in the future to weaken the distributed model even further, by working in failure-prone and asynchronous settings, in the hope of eventually getting practical distance sketches with provable performance guarantees.


%

\bibliographystyle{alpha}
\bibliography{Distance-Oracles}

\newcommand{\etalchar}[1]{$^{#1}$}
\begin{thebibliography}{KKM{\etalchar{+}}08}

\bibitem[ABC{\etalchar{+}}05]{ABCDGKNS05}
Ittai Abraham, Yair Bartal, T-H.~Hubert Chan, Kedar~Dhamdhere Dhamdhere, Anupam
  Gupta, Jon Kleinberg, Ofer Neiman, and Aleksandrs Slivkins.
\newblock Metric embeddings with relaxed guarantees.
\newblock In {\em Proceedings of the 46th Annual IEEE Symposium on Foundations
  of Computer Science}, FOCS '05, pages 83--100, Washington, DC, USA, 2005.
  IEEE Computer Society.

\bibitem[AKL{\etalchar{+}}79]{AKL+79}
Romas Aleliunas, Richard~M. Karp, Richard~J. Lipton, L{\'a}szl{\'o} Lov{\'a}sz,
  and Charles Rackoff.
\newblock Random walks, universal traversal sequences, and the complexity of
  maze problems.
\newblock In {\em FOCS}, pages 218--223, 1979.

\bibitem[ALPH01]{AHLP01}
Lada~A. Adamic, Rajan~M. Lukose, Amit~R. Puniyani, and Bernardo~A. Huberman.
\newblock Search in power-law networks.
\newblock {\em Physical Review}, 64, 2001.

\bibitem[BAS04]{BAS04}
Ashwin~R. Bharambe, Mukesh Agrawal, and Srinivasan Seshan.
\newblock Mercury: supporting scalable multi-attribute range queries.
\newblock In {\em SIGCOMM}, pages 353--366, 2004.

\bibitem[BBF04]{BBF04}
Thibault Bernard, Alain Bui, and Olivier Flauzac.
\newblock Random distributed self-stabilizing structures maintenance.
\newblock In {\em ISSADS}, pages 231--240, 2004.

\bibitem[CDG06]{CDG06}
T.-H.~Hubert Chan, Michael Dinitz, and Anupam Gupta.
\newblock Spanners with slack.
\newblock In {\em Proceedings of the 14th European Symposium on Algorithms},
  pages 196--207, 2006.

\bibitem[CFI{\etalchar{+}}09]{CFIKP09}
Reuven Cohen, Pierre Fraigniaud, David Ilcinkas, Amos Korman, and David Peleg.
\newblock Labeling schemes for tree representation.
\newblock {\em Algorithmica}, 53(1):1--15, 2009.

\bibitem[Coo05]{C05}
Brian~F. Cooper.
\newblock Quickly routing searches without having to move content.
\newblock In {\em IPTPS}, pages 163--172, 2005.

\bibitem[CTW93]{CTW93}
Don Coppersmith, Prasad Tetali, and Peter Winkler.
\newblock Collisions among random walks on a graph.
\newblock {\em SIAM J. Discret. Math.}, 6(3):363--374, 1993.

\bibitem[DCKM04]{Vivaldi}
Frank Dabek, Russ Cox, Frans Kaashoek, and Robert Morris.
\newblock Vivaldi: A decentralized network coordinate system.
\newblock In {\em Proceedings of the {ACM} {SIGCOMM} '04 Conference}, Portland,
  Oregon, August 2004.

\bibitem[DSW06]{DSW06}
Shlomi Dolev, Elad Schiller, and Jennifer~L. Welch.
\newblock Random walk for self-stabilizing group communication in ad hoc
  networks.
\newblock {\em IEEE Trans. Mob. Comput.}, 5(7):893--905, 2006.
\newblock also in PODC'02.

\bibitem[GKM03]{GKM03}
Ayalvadi~J. Ganesh, Anne-Marie Kermarrec, and Laurent Massouli\'{e}.
\newblock Peer-to-peer membership management for gossip-based protocols.
\newblock {\em IEEE Trans. Comput.}, 52(2):139--149, 2003.

\bibitem[GMS05]{GMS05}
Christos Gkantsidis, Milena Mihail, and Amin Saberi.
\newblock Hybrid search schemes for unstructured peer-to-peer networks.
\newblock In {\em INFOCOM}, pages 1526--1537, 2005.

\bibitem[GPPR04]{GPPR04}
Cyril Gavoille, David Peleg, St{\'e}phane P{\'e}rennes, and Ran Raz.
\newblock Distance labeling in graphs.
\newblock {\em J. Algorithms}, 53(1):85--112, 2004.

\bibitem[IJ90]{IJ90}
Amos Israeli and Marc Jalfon.
\newblock Token management schemes and random walks yield self-stabilizing
  mutual exclusion.
\newblock In {\em PODC}, pages 119--131, 1990.

\bibitem[KKD01]{KKD01}
David Kempe, Jon~M. Kleinberg, and Alan~J. Demers.
\newblock Spatial gossip and resource location protocols.
\newblock In {\em STOC}, pages 163--172, 2001.

\bibitem[KKKP04]{KKKP04}
Michal Katz, Nir~A. Katz, Amos Korman, and David Peleg.
\newblock Labeling schemes for flow and connectivity.
\newblock {\em SIAM J. Comput.}, 34(1):23--40, 2004.

\bibitem[KKM{\etalchar{+}}08]{khan-podc}
M.~Khan, F.~Kuhn, D.~Malkhi, G.~Pandurangan, and K.~Talwar.
\newblock Efficient distributed approximation algorithms via probabilistic tree
  embeddings.
\newblock In {\em Proc. 27th ACM Symp. on Principles of Distributed Computing
  (PODC)}, 2008.

\bibitem[Kle00]{K00}
Jon~M. Kleinberg.
\newblock The small-world phenomenon: an algorithm perspective.
\newblock In {\em STOC}, pages 163--170, 2000.

\bibitem[KP08]{khan-disc}
M.~Khan and G.~Pandurangan.
\newblock A fast distributed approximation algorithm for minimum spanning
  trees.
\newblock {\em Distributed Computing}, 20:391--402, 2008.

\bibitem[KR04]{KR04}
David~R. Karger and Matthias Ruhl.
\newblock Simple efficient load balancing algorithms for peer-to-peer systems.
\newblock In {\em SPAA}, pages 36--43, 2004.

\bibitem[LCC{\etalchar{+}}02]{LCCLS02}
Qin Lv, Pei Cao, Edith Cohen, Kai Li, and Scott Shenker.
\newblock Search and replication in unstructured peer-to-peer networks.
\newblock In {\em ICS}, pages 84--95, 2002.

\bibitem[LKRG03]{LKRG03}
Dmitri Loguinov, Anuj Kumar, Vivek Rai, and Sai Ganesh.
\newblock Graph-theoretic analysis of structured peer-to-peer systems: routing
  distances and fault resilience.
\newblock In {\em SIGCOMM}, pages 395--406, 2003.

\bibitem[LS03]{LawS03}
Ching Law and Kai-Yeung Siu.
\newblock Distributed construction of random expander networks.
\newblock In {\em INFOCOM}, 2003.

\bibitem[MG07]{MG07}
Rams{\'e}s Morales and Indranil Gupta.
\newblock Avmon: Optimal and scalable discovery of consistent availability
  monitoring overlays for distributed systems.
\newblock In {\em ICDCS}, page~55, 2007.

\bibitem[MR95]{MR95}
Rajeev Motwani and Prabhakar Raghavan.
\newblock {\em Randomized algorithms}.
\newblock Cambridge University Press, New York, NY, USA, 1995.

\bibitem[Pel00]{peleg}
David Peleg.
\newblock {\em Distributed computing: a locality-sensitive approach}.
\newblock Society for Industrial and Applied Mathematics, Philadelphia, PA,
  USA, 2000.

\bibitem[PK09]{PK09}
Gopal Pandurangan and Maleq Khan.
\newblock Theory of communication networks.
\newblock In {\em Algorithms and Theory of Computation Handbook, Second
  Edition}. CRC Press, 2009.

\bibitem[Sli05]{S05b}
Aleksandrs Slivkins.
\newblock Distance estimation and object location via rings of neighbors.
\newblock In {\em Proceedings of the twenty-fourth annual ACM symposium on
  Principles of distributed computing}, PODC '05, pages 41--50, New York, NY,
  USA, 2005. ACM.

\bibitem[Sli07]{S07}
Aleksandrs Slivkins.
\newblock Towards fast decentralized construction of locality-aware overlay
  networks.
\newblock In {\em Proceedings of the twenty-sixth annual ACM symposium on
  Principles of distributed computing}, PODC '07, pages 89--98, New York, NY,
  USA, 2007. ACM.

\bibitem[TZ05]{TZ05}
Mikkel Thorup and Uri Zwick.
\newblock Approximate distance oracles.
\newblock {\em J. ACM}, 52(1):1--24, 2005.

\bibitem[WSS05]{WSS05}
Bernard Wong, Aleksandrs Slivkins, and Emin~G\"{u}n Sirer.
\newblock Meridian: a lightweight network location service without virtual
  coordinates.
\newblock In {\em Proceedings of the 2005 conference on Applications,
  technologies, architectures, and protocols for computer communications},
  SIGCOMM '05, pages 85--96, New York, NY, USA, 2005. ACM.

\bibitem[ZS06]{ZS06}
Ming Zhong and Kai Shen.
\newblock Random walk based node sampling in self-organizing networks.
\newblock {\em Operating Systems Review}, 40(3):49--55, 2006.

\bibitem[ZSS05]{ZSS05}
Ming Zhong, Kai Shen, and Joel~I. Seiferas.
\newblock Non-uniform random membership management in peer-to-peer networks.
\newblock In {\em INFOCOM}, pages 1151--1161, 2005.

\end{thebibliography}

\end{document}